\documentclass[journal,twoside,web]{ieeecolor}

\usepackage{generic}
\usepackage{lcsys}

\usepackage[utf8]{inputenc}
\usepackage[english]{babel}

\usepackage{amsmath,amssymb,amsfonts}
\usepackage{mathtools}
\usepackage{mathrsfs}
\usepackage{bm}
\usepackage{calrsfs}

\usepackage{amsthm}

\usepackage{graphicx}
\graphicspath{{images/}}

\usepackage[caption=false,font=footnotesize]{subfig}

\usepackage{booktabs}
\usepackage{array}
\usepackage{longtable}
\usepackage{multicol}

\usepackage{algorithmic}

\usepackage{textcomp}
\usepackage[version=4]{mhchem}
\usepackage{siunitx}
\usepackage{nomencl}
\usepackage{enumerate}

\usepackage{optidef}

\usepackage{cite}

\newcommand{\comment}[1]{}

\newcommand{\be}{\begin{equation}}
\newcommand{\ee}{\end{equation}}
\newcommand{\ben}{\begin{equation*}}
\newcommand{\een}{\end{equation*}}
\newcommand{\bena}{\begin{eqnarray*}}
\newcommand{\eena}{\end{eqnarray*}}
\newcommand{\bea}{\begin{eqnarray}}
\newcommand{\eea}{\end{eqnarray}}

\DeclareMathAlphabet{\mathcal}{OMS}{cmsy}{m}{n}
\theoremstyle{definition}
\newtheorem{theorem}{Theorem}[section]

\newtheorem{definition}[theorem]{Definition}
\newtheorem{assumption}[theorem]{Assumption}
\newtheorem{remark}[theorem]{Remark}

\begin{document}

\def\BibTeX{{\rm B\kern-.05em{\sc i\kern-.025em b}\kern-.08em
    T\kern-.1667em\lower.7ex\hbox{E}\kern-.125emX}}
\markboth{\journalname, VOL. XX, NO. XX, XXXX 2017}
{Author \MakeLowercase{\textit{et al.}}: Preparation of Papers for IEEE Control Systems Letters (August 2022)}

\title{Robust Safety Filters for Lipschitz-Bounded Adaptive Closed-Loop Systems with Structured Uncertainties}

\author{Johannes Autenrieb$^{1}$, Peter A. Fisher$^{2}$ and Anuradha Annaswamy$^{2}$
\thanks{$^{1}$ German Aerospace Center (DLR), Institute of Flight Systems, Department of Flight Dynamics and Simulation, 38108, Braunschweig, Germany.
(email: \texttt{johannes.autenrieb@dlr.de})}
\thanks{$^{2}$ Department of Mechanical Engineering, Massachusetts Institute of Technology, Cambridge, MA 02139, USA.
(email: \texttt{pafisher@mit.edu, aanna@mit.edu}) 
\newline The second author would like to acknowledge the support of the Boeing Strategic University Initiative.}
}

\maketitle
\thispagestyle{empty}

\begin{abstract}
Adaptive control provides closed-loop stability and reference tracking for uncertain dynamical systems through online parameter adaptation. These properties alone, however, do not ensure safety in the sense of forward invariance of state constraints, particularly during transient phases of adaptation. Control barrier function (CBF)–based safety filters have been proposed to address this limitation, but existing approaches often rely on conservative constraint tightening or static safety margins within quadratic program formulations. This paper proposes a reference-based adaptive safety framework for systems with structured parametric uncertainty that explicitly accounts for transient plant–reference mismatch. Safety is enforced at the reference level using a barrier-function-based filter, while adaptive control drives the plant to track the safety-certified reference. By exploiting Lipschitz bounds on the closed-loop tracking error dynamics, a tracking-error-dependent robust CBF condition is derived and equivalently reformulated as a convex second-order cone program (SOCP). The proposed safety-filter formulation reduces conservatism relative to fixed-margin CBF formulations by rendering the resulting safety constraints progressively less restrictive as the plant–reference tracking error decreases, while preserving formal guarantees of forward invariance and closed-loop stability.

\end{abstract}


\section{Introduction}
\label{Introduction}
The field of adaptive control has long focused on providing real-time control inputs for dynamic systems with uncertain parameters through parameter learning and control design based on stability analysis \cite{Narendra2005, Ioannou1996, Sastry_1989, Slotine1991, Krstic1995, Astrom_1995}. Using Lyapunov analysis, continuity arguments, and reference model formulations, adaptive control provides guarantees of closed-loop stability and asymptotic tracking despite parametric uncertainties.

In parallel, a growing body of research has emerged in the area of safety-critical systems \cite{ames2016control, nguyen2016exponential, xiao2019control}, motivated by the need to provide verifiable guarantees of safe behavior in systems with increasing levels of autonomy and interaction with uncertain environments. A central tool in this context is the concept of CBFs, which enforce forward invariance of a prescribed safe set by imposing inequality constraints on the system dynamics.

\begin{figure}
\centering
\includegraphics[width=\linewidth]{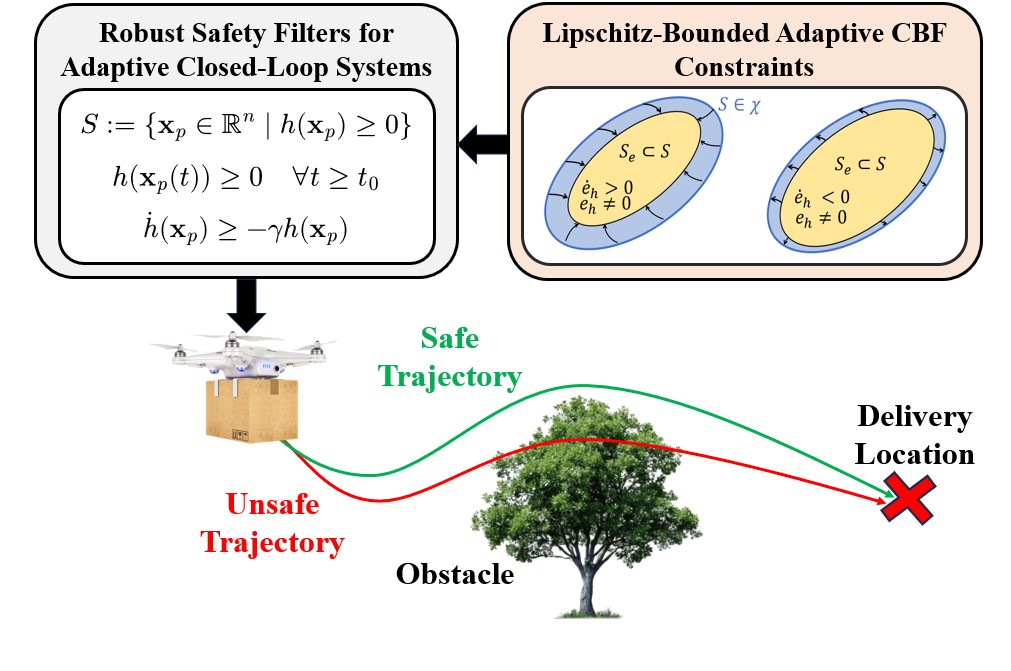}
\caption{Illustration of the proposed robust safety-critical adaptive control concept for systems with structured uncertainties.}
\label{fig:overview_concept}
\end{figure}

While adaptive control and CBF-based safety filtering are well established individually, integrating these frameworks presents fundamental challenges. Standard safety filters typically rely on accurate knowledge of the system dynamics. As a result, when applied to uncertain systems, constraint violations may occur during transient phases of learning. Consequently, several approaches have been proposed to address safety for systems with parametric uncertainty using CBF-based techniques \cite{Xiangru_2015, Jankovic_2018, Buch_2022, Taylor_2020_ml, Ames_2020a}. In \cite{Ames_2020a}, adaptive CBF conditions are formulated by embedding unknown parameters directly into the barrier constraints and requiring these conditions to hold uniformly over sets of admissible parameters and adaptation gains. While this approach provides formal safety guarantees, it represents model uncertainty primarily through structured drift dynamics, while the control effectiveness matrix is typically assumed to be known. Extensions that account for uncertainty in the input matrix have recently been considered in \cite{Wang2024}. Nevertheless, even in these formulations the resulting invariance conditions may become restrictive. To reduce conservatism, \cite{Lopez2021} introduces robust adaptive control barrier functions, which relax the invariance conditions while preserving safety guarantees.

Despite these developments, most existing approaches formulate safety conditions directly at the plant level and treat adaptation within the barrier constraints themselves. As a consequence, these formulations are not naturally integrated with classical model reference adaptive control architectures that are widely used to provide stability and performance guarantees. In earlier work \cite{Autenrieb2023}, a reference-based safety architecture was introduced in which a CBF safety filter certifies a safe reference command for a known reference model, while an adaptive controller drives the uncertain plant to track this certified reference. In that work, robustness against transient plant–reference mismatch during adaptation was addressed using heuristically chosen error-based robustness margins.

In this paper, we extend this concept to systems with structured parametric uncertainties affecting both the plant dynamics and the actuator effectiveness. The proposed architecture is illustrated in Fig.~\ref{fig:overview_concept}. A CBF-based safety filter generates a certified reference command for a known reference model, while an adaptive controller drives the uncertain plant to track this certified reference. Unlike the earlier formulation, where robustness against plant–reference mismatch during adaptation was handled through heuristic margins, the proposed approach incorporates explicit worst-case conditions on the plant dynamics. By exploiting Lipschitz bounds that relate barrier-function mismatch to the closed-loop tracking error, a robust barrier condition is obtained that captures plant–reference safety mismatch during learning. The resulting constraint depends on norms of the reference input and naturally leads to a convex SOCP formulation for the safety filter, providing a tractable mechanism for enforcing forward invariance while reducing conservatism and yielding smoother transient behavior compared with existing safety filters with formal guarantees.

\section{Preliminaries \& Problem Formulation}
\label{Preliminaries}
\subsection{Preliminaries}

We consider a nonlinear control-affine system of the form
\begin{equation}
\label{NonlinearPlant}
\dot{x} = f(x) + g(x)u
\end{equation}
where $x \in \chi \subset \mathbb{R}^n$ denotes the system state, $u \in \mathbb{R}^m$ the control input, and the vector fields $f: \chi \rightarrow \mathbb{R}^n$ and $g: \chi \rightarrow \mathbb{R}^{n \times m}$ are assumed to be locally Lipschitz continuous. In order to define safety, we consider a continuously differentiable function  $h: \chi \rightarrow \mathbb{R}$ and a set $S$ defined as the zero-superlevel set of $h$, yielding:
\begin{equation}
    \label{Safe_set_1}
    S \triangleq \begin{Bmatrix} x \in  \chi | h(x) \geq 0 \end{Bmatrix},
\end{equation}
\begin{equation}
    \label{Safe_set_2}
    \partial S \triangleq \begin{Bmatrix} x \in  \chi | h(x) = 0 \end{Bmatrix},
\end{equation}
\begin{equation}
    \label{Safe_set_3}
    int(S) \triangleq \begin{Bmatrix} x \in  \chi | h(x) > 0 \end{Bmatrix}.
\end{equation}

\begin{definition}
The set $S$ is forward invariant for the system \eqref{NonlinearPlant}, if for every $x_0 \in S$, it follows $x \in S$ for $x(0) = x_0$ and all $t \in I(x_0) = [0,\tau_{max} = \infty)$ \cite{Nagumo_1942, Blanchini_1999}.
\end{definition}

We now introduce the notion of a CBF. The existence of such a function provides a sufficient condition for rendering the set (S) forward invariant \cite{Ames_2014,Ames_2017}.
\begin{definition}
Let $S \subset \chi$ be the zero-superlevel set of a continuously differentiable function $h: \chi \rightarrow \mathbb{R}$. The function $h$ is a CBF for $S$ if there exists a class $\mathcal{K}_{\infty}$ function $\alpha(h(x))$ such that for the system defined in \eqref{NonlinearPlant} we obtain:
\begin{equation}
    \label{eq:ControlBarrierFunction}
    \sup_{u\in \mathbb{R}^m} \frac{\partial h}{\partial x}\begin{bmatrix} f(x) + g(x)u \end{bmatrix}  \geq -\alpha(h(x)),
\end{equation}
for all $x \in S$.
\end{definition}

\begin{theorem}
\label{theorem_LCBF}
Given a set \( S \subset \chi \), defined via the associated CBF as in \eqref{Safe_set_1}, any Lipschitz continuous controller \( {k}({x}) \in K_{S}({x}) \) with 
\begin{equation}
    K_{S} ({x}) = \big\{ {u} \in U : L_{{f}} h({x}) + L_{{g}} h({x}) {u} + \alpha(h({x})) \ge 0 \big\}
    \label{definition_safe_controller}
\end{equation}
renders the system \eqref{NonlinearPlant} forward invariant within \( S \) \cite{XU2015}. 
\end{theorem}
One way to construct a controller satisfying \eqref{definition_safe_controller} is through a safety filter formulated as a quadratic program (QP), as proposed in~\cite{Ames_2014}: 
\begin{equation}
\begin{aligned}
{u} = \arg\min_{{u}\in U} \quad
& \|{u} - {u}^\star\|_2^2 \\
\text{s.t.}\quad
& L_{{f}} h({x})
  + L_{{g}} h({x})\,{u}
  \ge -\alpha\!\left(h({x})\right),
\end{aligned}
\label{eq:qp_cbf}
\end{equation}
where ${u}^\star$ is a performance-oriented, but potentially not safe, control input.

\subsection{Problem Formulation}
\label{Problem Formulation}
We consider the following linear time-invariant (LTI) system representing a plant with partially unknown dynamics:
\begin{equation}
    \dot{\bm{x}}_p = A_p \bm{x}_p + B_p \Lambda \bm{u},
    \label{eq:LinearPlantModel_Problem}
\end{equation}
where \(\bm{x}_p \in \mathbb{R}^n\) is the plant state vector, \(\bm{u} \in \mathbb{R}^m\) is the control input, \(A_p \in \mathbb{R}^{n \times n}\) is an unknown but assumed constant system matrix, \(B_p \in \mathbb{R}^{n \times m}\) is a known full-rank input matrix and $\Lambda \in \mathbb{R}^{m \times m}$ is assumed to be unknown, with $\Lambda$ being diagonal and strictly positive constant matrix.

The objective of the proposed control design is to determine a $u$ for \eqref{eq:LinearPlantModel_Problem} such that the plant state $x_p$ tracks a desired reference $ x_d$ and that for any initial condition $x_0 := x(t_0) \in S$, it is ensured that the plant state vector $x_p$ stays within the safe set $S \in \mathbf{R}^n$. Equivalently, the control input must ensure that the CBF condition is satisfied, i.e., \(h(x_p(t)) \geq 0\) for all \(t \geq t_0\).

\section{Adaptive Closed-Loop System Design}
\label{Safe_OMRAC}
The following assumptions are made regarding the unknown parameters of the considered plant in \eqref{eq:LinearPlantModel_Problem}:
\begin{assumption}
\label{assumption_matching_condition}
Constant matrices $\Theta_x^*$ and $\Theta_r^*$ exist that solve the following:
\begin{equation}
\label{matching_condition_1}
A_m = A_p + B_p \Lambda \Theta_x^*,
\end{equation}
\begin{equation}
\label{matching_condition_2}
B_m = B_p \Lambda \Theta_r^*.
\end{equation}
\end{assumption}
\begin{assumption}
The input uncertainty $\Lambda$ is a diagonal positive definite matrix of the form:
\begin{equation}
    \label{Lambda_assumption}
    \Lambda =
    \begin{bmatrix}
        \lambda_1   & \hdots    & 0\\
        \vdots      & \ddots    & 0\\
        0           & \hdots    & \lambda_m
    \end{bmatrix}.
\end{equation}
\end{assumption}

The objective of the adaptive control design is to enforce stable tracking of a prescribed reference model by a linear plant with unknown \(A_p\) and \(\Lambda\). To define the desired closed-loop behavior, we consider the following reference model:
\begin{equation}
    \dot{{x}}_m = A_m {x}_m + B_m {r}.
    \label{eq:reference_model}
\end{equation}
We consider the following adaptive controller:
\begin{equation}
    {u} = \widehat{\Theta}_x {x}_p + \widehat{\Theta}_r {r},
    \label{eq:adaptive_control_law}
\end{equation}
where \(\widehat{\Theta}_x \in \mathbb{R}^{m \times n}\) is the current estimate of the feedback gain matrix, \(\widehat{\Theta}_r \in \mathbb{R}^{m \times p}\) is the current estimate of the feedforward gain matrix, and \({r} \in \mathbb{R}^p\) is a desired reference input signal. We define the tracking error between the plant in \eqref{eq:LinearPlantModel_Problem} and the reference model in \eqref{eq:reference_model} by:
\begin{equation}
    e_x(t)=x_p(t)-x_m(t)
    \label{eq:tracking_error}.
\end{equation}

The time-varying parameters in \eqref{eq:adaptive_control_law} are adjusted using the following adaptation laws: 
\begin{equation}
\label{theta_x_adaption_law}
    \dot{\widehat{\Theta}}_x(t) = - \Gamma_x x_p(t) e_x(t)^\top P B_p\,, \,\,\,\,\,\, \Gamma_x>0
\end{equation}
\begin{equation}
\label{theta_r_adaption_law}
    \dot{\widehat{\Theta}}_r(t) = - \Gamma_r r(t) e_x(t)^\top P B_p\,, \,\,\,\,\,\, \Gamma_r>0
\end{equation}
where $P$ is the solution of the Lyapunov equation $A_m^\top P + P A_m = -Q$, where $Q>0$. 

Further, we define the parameter estimation errors:
\begin{align}
    \widetilde{\Theta}_x &:= \Theta_x^* - \widehat{\Theta}_x, \label{eq:gain_error_x}\\
    \widetilde{\Theta}_r &:= \Theta_r^* - \widehat{\Theta}_r, \label{eq:gain_error_r}
\end{align} 
as well as a corresponding output error $e_u(t)=u(t)-u^*(t)$, which will be useful to quantify the safety of the adaptive controller, where $u^*(t)$ represents the ideal control input and is defined as:
\begin{equation}
\label{eq:idealu}
u^*(t)  = \Theta_x^* x_p(t) + \Theta_r^* r_s(t).
\end{equation}

\begin{theorem}
\label{Theorem_Basic_Adaptive_Controller}
The overall closed-loop adaptive system defined by the plant in \eqref{eq:LinearPlantModel_Problem}, the reference model in \eqref{eq:reference_model}, the control input in \eqref{eq:adaptive_control_law}, and the adaptation laws in \eqref{theta_x_adaption_law} and \eqref{theta_r_adaption_law} has globally bounded solutions for any initial conditions $x_p(t_0)$, $x_m(t_0)$, $\widehat\Theta_x(t_0)$, and $\widehat\Theta_r(t_0)$. Moreover, if the reference signal $r(t)$ is bounded, both the errors $e_x$ and $e_u(t)$ converge to zero as $t \rightarrow \infty$.
\end{theorem}
\begin{proof}
The proof follows from standard adaptive control arguments, since the error dynamics is of the form
\begin{equation}
    \label{error_dyanmics_1}
    \dot{e}_x = A_m e_x + B_p \Lambda (\widetilde{\Theta}_x x_p + \widetilde{\Theta}_r r_s).
\end{equation} Taking into account \eqref{eq:gain_error_x} and \eqref{eq:gain_error_r}, we consider the following Lyapunov function candidate
\begin{equation}
\begin{split}    
    \label{Lyapunov_function1}
    V = &\frac{1}{2} e_x^\top P e_x + \frac{1}{2} \operatorname{Tr} [ \widetilde{\Theta}_x \Gamma_x^{-1} \widetilde{\Theta}_x^\top \Lambda ] \\
   + &\frac{1}{2} \operatorname{Tr}  [ \widetilde{\Theta}_r\Gamma_r^{-1} \widetilde{\Theta}_r^\top \Lambda ].
\end{split}
\end{equation}
Using \eqref{theta_x_adaption_law}, \eqref{theta_r_adaption_law}, and \eqref{error_dyanmics_1}, straightforward algebra yields $\dot{V} = -e_x^\top Qe_x$. Since $\dot{V} \leq 0$, we immediately obtain $e_x, \hat{\Theta}_x, \hat{\Theta}_r \in \mathcal{L}_\infty$, and boundedness of $x_p$ follows from boundedness of $r$ and thus boundedness of $x_m$. Additionally, the form of $\dot{V}$ gives us $e_x \in \mathcal{L}_2$. Now, as $e_x$ is bounded and has a bounded derivative, an application of the Barbalat's Lemma leads to $\lim_{t \to \infty} e_x (t) = 0$ \cite{Narendra2005}. From \eqref{error_dyanmics_1}, it follows that $e_u(t)$ is an input into an LTI system, with a bounded derivative, whose state is $e_x(t)$; therefore it follows that $\lim_{t \to \infty} e_u (t) = 0$. {\color{black}We refer to~\cite{Narendra2005,Ioannou1996,Sastry_1989} for further details of the proof.
}
\end{proof}

\section{Robust Safety for Lipschitz-Bounded Adaptive Closed-Loop Systems with Structured Uncertainties}
\label{Safe_CCRM}
The adaptive control design from the previous section guarantees closed-loop stability and asymptotic tracking of the reference model \eqref{eq:reference_model} by the uncertain plant \eqref{eq:LinearPlantModel_Problem}, but does not ensure forward invariance for a defined safe set (i.e., satisfaction of state constraints for all \(t \ge t_0\)). Since the standard CBF condition \eqref{eq:ControlBarrierFunction} depends explicitly on the unknown plant dynamics, direct plant-level safety enforcement is not feasible. This motivates the introduction of a safety mechanism compatible with the adaptive control architecture. We define the safe set as
\begin{equation}
    S := \left\{ {x}_p \in \chi \;\middle|\; h({x}_p) \geq 0 \right\}.
    \label{eq:safe_set_thm_min}
\end{equation}
Safety of \eqref{eq:LinearPlantModel_Problem} is achieved if the condition
\begin{equation}
    h({x}_p(t)) \geq 0 \quad \forall t \geq t_0
    \label{eq:safety_requirement}
\end{equation}
is satisfied. A standard sufficient condition \cite{ames2016control}, to ensure forward invariance of \(S\) for the plant \eqref{eq:LinearPlantModel_Problem} is
\begin{equation}
    \dot{h}({x}_p) \geq -\gamma h({x}_p),
    \label{eq:standard_CBF_condition_plant}
\end{equation}
with \(\gamma > 0\). For the uncertain plant, however, the matrices \(A_p\) and \(\Lambda\) are unknown, and \(\dot{h}({x}_p)\) cannot be evaluated exactly. As a result, the condition \eqref{eq:standard_CBF_condition_plant} cannot be exactly enforced at the plant level. To address this limitation, the indirect safety approach proposed in \cite{Autenrieb2023,Fisher2026} is adopted: safety is enforced on the known reference model, while the adaptive controller compensates for the mismatch between the plant and the reference model (cf. Fig.~\ref{fig:control_architecture}).

\begin{figure*}
    \centering
    \includegraphics[width=0.7\textwidth]{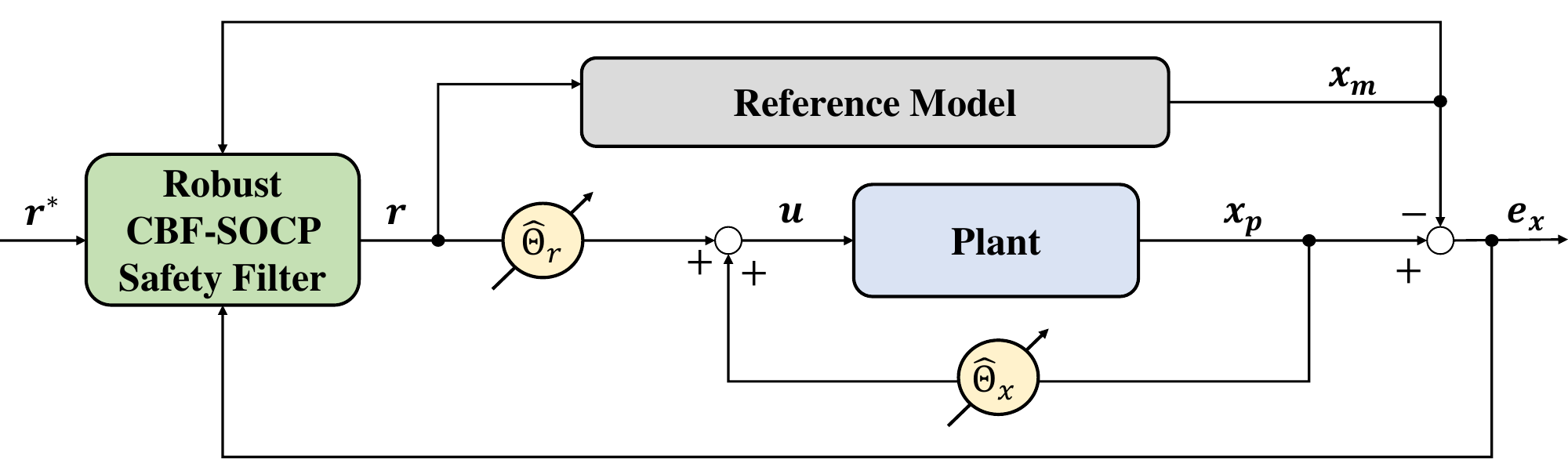}
    \caption{Concept of the proposed safe adaptive control architecture.}
    \label{fig:control_architecture}
\end{figure*}

To retain safety guarantees during transient adaptation, tracking errors and parameter mismatch must be explicitly accounted for in the safety filter design. To do so, we define the barrier function error
\begin{equation}
    e_h := h({x}_p) - h({x}_m).
    \label{eq:barrier_function_error}
\end{equation}

This error captures the deviation between the actual safety level of the plant and that of the reference model. We can express \eqref{eq:barrier_function_error} as:
\begin{equation}
    h({x}_p) = h({x}_m) + e_h.
\label{eq:barrier_function_value_plant_with_error}
\end{equation}

Differentiating \eqref{eq:barrier_function_value_plant_with_error} with respect to time yields:
\begin{equation}
    \dot{h}({x}_p) = \dot{h}({x}_m) + \dot{e}_h,
\label{eq:barrier_function_value_plant_with_error_diff}
\end{equation}
where
\begin{align}
    \dot{h}({x}_p) &= \left. \frac{\partial h}{\partial {x}} \right|_{{x}_p} \cdot \dot{{x}}_p 
    = \nabla h({x}_p)^\top (A_p {x}_p + B_p {u}), \\
    \dot{h}({x}_m) &= \left. \frac{\partial h}{\partial {x}} \right|_{{x}_m} \cdot \dot{{x}}_m 
    = \nabla h({x}_m)^\top (A_m {x}_m + B_m {r}). \label{eq:hdot_xm}
\end{align}

\eqref{eq:standard_CBF_condition_plant}, \eqref{eq:barrier_function_value_plant_with_error} and \eqref{eq:barrier_function_value_plant_with_error_diff} allow us to obtain the following composite safety condition:
\begin{align}
    \dot{h}({x}_m) + \dot{e}_h &\geq -\gamma (h({x}_m) + e_h),
    \label{eq:composite_barrier_condition_1}
\end{align}
which can be reformulated to
\begin{align}
    \dot{h}({x}_m) + \dot{e}_h + \gamma e_h &\geq -\gamma h({x}_m).
    \label{eq:composite_barrier_condition_2}
\end{align}

This inequality expresses a sufficient condition for safety that depends only on the reference model and the mismatch \(e_h, \dot{e}_h\), which we now aim to bound. A more conservative version of  \eqref{eq:composite_barrier_condition_2} is given by:
\begin{equation}
    \dot{h}({x}_m) - | \dot{e}_h | - \gamma |e_h|
    \geq -\gamma h({x}_m).
    \label{eq:conservative_composite_barrier_condition}
\end{equation}

\begin{assumption}
\label{ass:barrier_mismatch_bounds_min}
Let $h:\mathbb{R}^n\to\mathbb{R}$ be continuously differentiable. 
Considering the tracking error in \eqref{eq:tracking_error} and the  barrier function error \eqref{eq:barrier_function_error}. There exist constants $L_1,L_2>0$ such that, for all $t\ge t_0$ in the region of interest,
\begin{equation}
    | e_h | = | h({x}_p) - h({x}_m) | \leq L_1 \| {e}_x \|,
    \label{eq:Lipschitzt_bound_state}
\end{equation}
\begin{equation}
    | \dot{e}_h | = | \dot{h}({x}_p) - \dot{h}({x}_m) | \leq L_2 \| \dot{{e}}_x \|.
    \label{eq:Lipschitzt_bound_dynamics}
\end{equation}
\end{assumption}

\begin{assumption}
\label{ass:bounded_uncertainties_min}
There exist known constants $\overline{\Theta}_x>0$, $\overline{\Theta}_r>0$, and $\overline{\Lambda}>0$ such that, for all $t\ge t_0$,
\begin{equation}
\|\widetilde{\Theta}_x\|\le \overline{\Theta}_x,\qquad
\|\widetilde{\Theta}_r\|\le \overline{\Theta}_r,\qquad
\|\Lambda\|\le \overline{\Lambda}.
\label{eq:bounded_uncertainties_min}
\end{equation}
\end{assumption}

We choose the reference input $r(t)$ at every $t \geq t_0$ such that the following inequality holds:
\begin{align}
\dot h( x_m(t)) \ge &-\gamma h( x_m(t))
+ \gamma L_1\| e_x(t)\| + L_2\|A_m\|\,\| e_x(t)\| \notag\\
&+ L_2\|B_p\|\,\overline{\Lambda}\Big(\overline{\Theta}_x\| x_p(t)\|+\overline{\Theta}_r\| r(t)\|\Big),
\label{eq:final_safety_condition_uncertainties}
\end{align}
where $\dot{h}(x_m)$ is given in \eqref{eq:hdot_xm}. We will discuss how to choose such an $r(t)$ in Section \ref{subsec:SOCP}. The following result provides our guarantee of safety resulting from satisfaction of \eqref{eq:final_safety_condition_uncertainties}:
\begin{theorem}
\label{thm:indirect_robust_safety_min}
Consider the plant in \eqref{eq:LinearPlantModel_Problem}, the reference model in \eqref{eq:reference_model}, the tracking error defined in \eqref{eq:tracking_error} and the safe set \eqref{eq:safe_set_thm_min} under Assumptions~\ref{assumption_matching_condition}--\ref{ass:bounded_uncertainties_min}. Let $x_p(t_0) \in \mathcal{S}$, and let the reference input $ r(t)$ be chosen such that \eqref{eq:final_safety_condition_uncertainties} holds for all $t \geq t_0$.
Then, we have $x_0(t) \in \mathcal{S}\ \forall t \geq t_0$.
\end{theorem}

\begin{proof}
\eqref{eq:conservative_composite_barrier_condition} is a sufficient condition for \eqref{eq:composite_barrier_condition_2}.
Applying Assumption~\ref{ass:barrier_mismatch_bounds_min} to \eqref{eq:conservative_composite_barrier_condition} yields
\begin{equation}
\dot h( x_m) \ge -\gamma h( x_m) + L_2\|\dot{ e}_x\| + \gamma L_1\| e_x\|.
\label{eq:cbf_with_bounds_pf_min}
\end{equation}
Additionally, using \eqref{error_dyanmics_1} and Assumption~\ref{ass:bounded_uncertainties_min}, we obtain
\begin{align}
    \|\dot{ e}_x\| &\le \|A_m\|\,\| e_x\| + \|B_p\|\,\|\Lambda\|\left(\|\widetilde{\Theta}_x\|\,\| x_p\| + \|\widetilde{\Theta}_r\|\,\| r\|\right) \notag \\
    &\le \|A_m\|\,\| e_x\| + \|B_p\|\,\overline{\Lambda}\left(\overline{\Theta}_x\| x_p\|+\overline{\Theta}_r\| r\|\right). \label{eq:exdot_bound_worst_pf_min}
\end{align}
Substituting \eqref{eq:exdot_bound_worst_pf_min} into \eqref{eq:cbf_with_bounds_pf_min} yields precisely \eqref{eq:final_safety_condition_uncertainties}. Hence,
\eqref{eq:final_safety_condition_uncertainties}
implies \eqref{eq:conservative_composite_barrier_condition}, which implies \eqref{eq:composite_barrier_condition_2}. Together with \eqref{eq:composite_barrier_condition_2}, this shows that the plant satisfies \eqref{eq:standard_CBF_condition_plant}. Finally, since $h( x_p(t_0))\ge 0$ by hypothesis and \eqref{eq:standard_CBF_condition_plant} holds, the scalar comparison lemma implies $h( x_p(t))\ge 0$ for all $t\ge t_0$, proving \eqref{eq:safety_requirement} \cite{Khalil_2002}.
\end{proof}

\subsection{SOCP-based safe reference signal generation} 
\label{subsec:SOCP}

Given a bounded nominal reference input $r^*$ - e.g., one designed for reference tracking - we would like to choose $r$ as close to $r^*$ as possible such that \eqref{eq:final_safety_condition_uncertainties} is satisfied. However, $r$ appears nonlinearly in \eqref{eq:final_safety_condition_uncertainties}, preventing a standard QP formulation as in \cite{Ames_2019}. In this section, we show how to choose $r$ as the solution to a second-order cone problem (SOCP). First, in order for an $r$ satisfying \eqref{eq:final_safety_condition_uncertainties} to exist, the following condition must hold:
\begin{assumption}
\label{ass:cbf_authority}
Let $E > 0$ be any quantity such that $\|e_x(t)\| \leq E\ \forall t \geq t_0$. Then, there exists a $d > 0$ such that, for every $x_m \in \mathcal{S}$ where $\|\frac{\partial h}{\partial x}|_{x_m} B_m\| \leq L_2\|B_p\|\,\overline{\Lambda}\,\overline{\Theta}_r + d$, the following inequality holds:
\begin{equation*}
\begin{aligned}
    \frac{\partial h}{\partial x}\Big|_{x_m}&A_mx_m + \gamma h(x_m) \geq \max_{x \in \mathbb{R}^n : \|x - x_m\| \leq E}\\
    &\left\{(\gamma L_1 + L_2\|A_m\|)\|x - x_m\| + L_2\|B_p\|\,\overline{\Lambda}\,\overline{\Theta}_x\|x\|\right\}.
\end{aligned}
\end{equation*}
\end{assumption}
\begin{remark}
    We know from Theorem \ref{Theorem_Basic_Adaptive_Controller} that a finite $E > 0$ exists. Assumption \ref{ass:cbf_authority} is the requirement that $r = 0$ satisfies \eqref{eq:final_safety_condition_uncertainties} wherever the control authority over $h$, quantified by $\|\frac{\partial h}{\partial x}|_{x_m}B_m\|$, is dominated by the uncertainty in the input matrix, quantified by $L_2\|B_p\|\,\overline{\Lambda}\,\overline{\Theta}_r$. This assumption is similar to Assumption 3.2 in \cite{Fisher2026}, and may be relaxed by modifying \eqref{eq:final_safety_condition_uncertainties} similarly to Section III.E in \cite{Fisher2026}.
\end{remark}

To obtain an SOCP, we first
%
introduce an auxiliary scalar $\eta \in \mathbb{R}_+$ as a conic upper bound on the reference norm,
\begin{equation}
\|{r}\| \leq \eta.
\label{eq:eta_definition}
\end{equation}
Substituting \eqref{eq:eta_definition} and \eqref{eq:hdot_xm} into \eqref{eq:final_safety_condition_uncertainties}
yields the
safety condition
\begin{equation}
a^\top  r + c\,\eta \ge \beta,
\label{eq:affine_safety_constraint}
\end{equation}
where
\begin{align}
a &:=
\left(
\left.\frac{\partial h}{\partial {x}}\right|_{{x}_m}
B_m
\right)^\top,
\label{eq:a_definition}
\\[4pt]
c &:=
L_2\|B_p\|\,\overline{\Lambda}\,\overline{\Theta}_r,
\label{eq:c_definition}
\\[4pt]
\beta
&:=
\left.\frac{\partial h}{\partial {x}}\right|_{{x}_m}
A_m  x_m
+ \gamma h( x_m) - \Delta( x_p, e_x),
\label{eq:beta_definition}\\[4pt]
\Delta({x}_p,{e}_x)
&:=
\gamma L_1\| e_x\|
+ L_2\|A_m\|\,\| e_x\|
\notag\\
&\quad
+ L_2\|B_p\|\,\overline{\Lambda}\,\overline{\Theta}_x\,\| x_p\|.
\label{eq:Delta_definition}
\end{align}

The goal of the safety filter is to determine a reference command $ r$
that remains as close as possible to a desired reference $ r^*$ while
penalizing the conservatism introduced by the auxiliary variable $\eta$.
This leads to the optimization problem
\begin{equation}
\min_{ r,\eta}\;\| r- r^*\| + \rho\,\eta,
\label{eq:socp_objective}
\end{equation}
where $\rho>0$ is a user-defined trade-off parameter balancing tracking
performance and conservatism. To express this problem in standard SOCP form,
we introduce an auxiliary scalar $v\in\mathbb{R}_+$ satisfying
\begin{equation}
\| r- r^*\| \le v,
\label{eq:epigraph_soc}
\end{equation}
which corresponds to an epigraph reformulation of the quadratic objective in \eqref{eq:socp_objective}.
The equivalent convex formulation therefore becomes
\begin{equation}
\min_{ r,\eta,v}\; v + \rho\,\eta .
\label{eq:socp_objective_linear}
\end{equation}

The reference-level safety filter can therefore be written as the following SOCP:
\begin{equation}
\begin{aligned}
\min_{ r,\eta,v} \quad &
v + \rho\,\eta \\[4pt]
\text{s.t.}\quad
& \| r- r^*\| \le v,\\
& \| r\| \le \eta,\\
& -a^\top  r - c\,\eta \le -\beta .
\end{aligned}
\label{eq:socp_filter}
\end{equation}

\section{Simulations}
\label{Simulation}
\begin{figure}
    \centering
    \includegraphics[width=\linewidth]{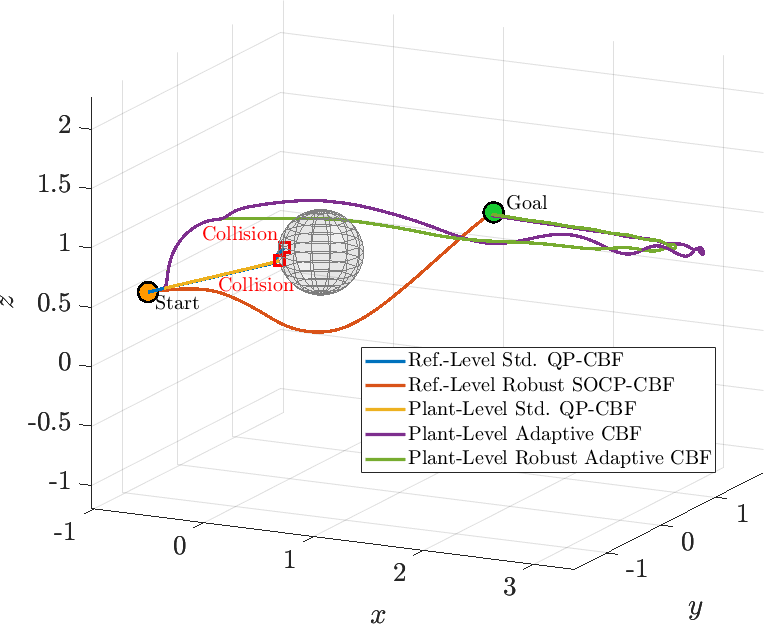}
    \caption{3D quadrotor trajectories for different safety filters under model mismatch.}
    \label{fig:results1}
\end{figure}

Numerical simulations were conducted in MATLAB to evaluate the proposed safety architecture on a quadrotor platform. An adjusted translational model of the quadrotor dynamics from \cite{Dydek12} has been used. The state vector consists of position and velocity components $x=[p_x,p_y,p_z,v_x,v_y,v_z]^T$, while the control inputs represent accelerations. For the considered simulation scenario, the quadrotor model is adjusted to fit the LTI system structure introduced in \eqref{eq:LinearPlantModel_Problem}. To evaluate robustness, the true plant differs from the nominal model through static deviations in both the drift dynamics and the actuator effectiveness. The simulated plant therefore uses $A_{p,\mathrm{sim}}=\Delta A_p A_p$ and $\Lambda_{\mathrm{sim}}\neq I$, representing uncertainty in both the unforced dynamics and the input effectiveness. The considered task is obstacle avoidance in three-dimensional space. The quadrotor is initialized at a start position and must reach a goal location while avoiding a static spherical obstacle located between the start and the goal. Safety is enforced through a CBF of the form $h(x)=\|p-p_o\|^2-r_o^2\ge0$, where $p$ denotes the quadrotor position, $p_o$ is the obstacle center, and $r_o$ its radius.

The simulation compares several safety enforcement strategies, including a nominal model-based standard plant-level QP-CBF \cite{Ames_2019}, adaptive CBF \cite{Ames_2020a}, robust adaptive CBF \cite{lopez2020adaptive}, and standard QP-CBF applied on the reference-level. In addition, the proposed reference-level robust SOCP-CBF is evaluated.
\begin{figure*}
    \centering
    \includegraphics[width=\linewidth]{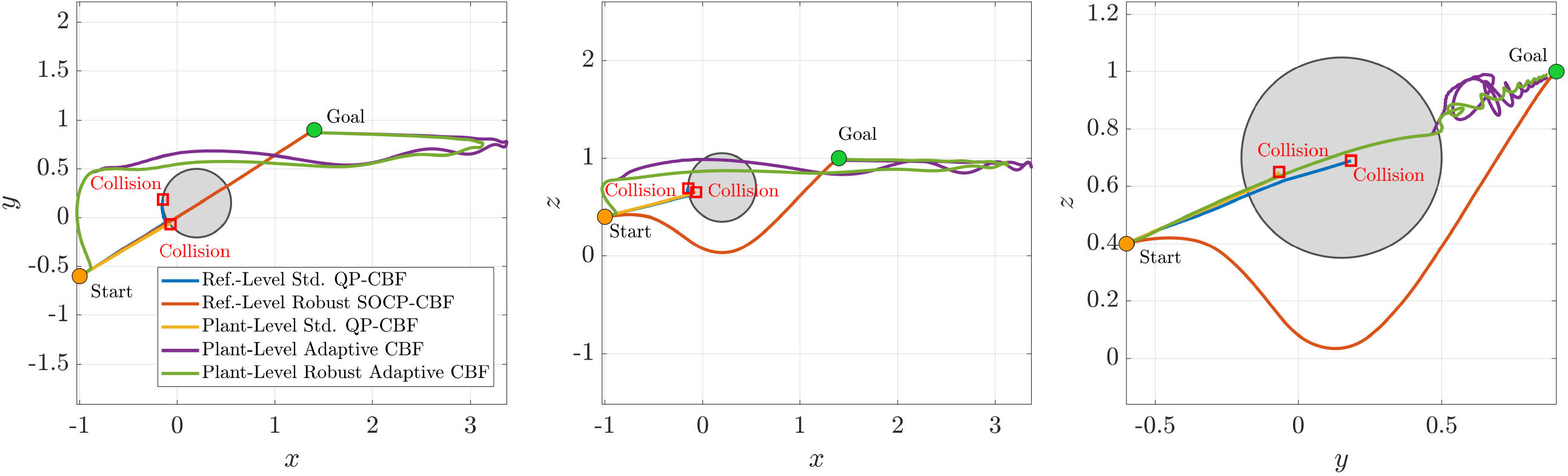}
    \caption{Planar projections of the trajectories for different safety filters under model mismatch.}
    \label{fig:results2}
\end{figure*}
Fig.~\ref{fig:results1} and Fig.~\ref{fig:results2} illustrate the resulting trajectories for the different approaches. As expected, the standard plant-level CBF and the reference-level CBF without robustness fail to guarantee safety under the considered model mismatch, leading to collisions with the obstacle. {\color{black}The adaptive CBF and robust adaptive CBF approaches maintain safety despite the plant uncertainties, but exhibit more oscillatory trajectories. This behavior can be attributed to the interaction between the adaptive controller, which adapts to improve tracking of the nominal reference, and the adaptive safety filter, which simultaneously modifies the effective control action to enforce safety under uncertain input effectiveness. In contrast, the proposed reference-level robust SOCP-CBF modifies the reference command before tracking, thereby avoiding this conflict between tracking and safety enforcement. As a result, it maintains safety while producing smoother trajectories that avoid the obstacle with a consistent safety margin and converge efficiently toward the goal.} 
Overall, the simulation results demonstrate that the proposed SOCP-based reference-level safety filter provides a robust mechanism for enforcing safety under structured plant uncertainty while maintaining smooth closed-loop behavior and efficient goal convergence..

\section{Conclusions}
\label{Conclusions}
This paper proposes a reference-based safety architecture in which a CBF-based safety filter certifies safe reference commands for a known reference model, while an adaptive controller drives the uncertain plant to track the certified trajectory. The key contribution is a robust reference-level CBF condition that explicitly accounts for a transient plant–reference mismatch using Lipschitz bounds, resulting in a safety constraint whose conservatism decreases as tracking improves. To handle the resulting norm-dependent robustness terms efficiently, the safety filter is formulated as a convex SOCP using an auxiliary conic bound and a penalty on conservatism, yielding a least-invasive certified command. Smulation-based comparisons using a quadcopter obstacle-avoidance scenario were conducted to evaluate adaptive and non-adaptive safety-filter approaches under structured model mismatch conditions. The results show that the proposed SOCP-based filter maintains safety while producing smoother trajectories. Overall, the framework provides a tractable mechanism for enforcing forward invariance in adaptive closed-loop systems under structured uncertainty.

\bibliographystyle{IEEEtran}
\bibliography{./Bibliography/references} 

\end{document}